\documentclass[conference,10pt]{IEEEtran}
\IEEEoverridecommandlockouts
\usepackage{cite}

\usepackage{amsfonts,amsmath,amssymb,amsthm,mathrsfs,bm,bbm}
\usepackage{graphicx}
\usepackage{overpic}
\usepackage{textcomp}
\usepackage{xcolor}
\usepackage{flushend}
\usepackage{enumitem}
\usepackage{algpseudocode}

\ifCLASSOPTIONcompsoc 
	\usepackage[caption=false,font=normalsize,labelfon t=sf,textfont=sf]{subfig} 
\else 
	\usepackage[caption=false,font=footnotesize]{subfig} \fi 
    
\newcommand{\eqdef}{:=}

\renewcommand{\vec}[1]{\bm{#1}}		
\newcommand{\rvec}[1]{\mathbbm{#1}} 		
\newcommand{\E}{\mathsf{E}}		
\newcommand{\Var}{\mathsf{V}}			
\newcommand{\stdset}[1]{\mathbbmss{#1}}	
\newcommand{\set}[1]{\mathcal{#1}}		

\newcommand{\CN}{\mathcal{CN}}			
\newcommand{\herm}{\mathsf{H}}			

\newtheorem{proposition}{Proposition}

\newtheorem{remark}{Remark}

\begin{document}

\title{Two-timescale weighted sum-rate maximization for large cellular and cell-free massive MIMO 
}

\author{
\IEEEauthorblockN{Lorenzo Miretti$^{1}$, Emil Björnson$^2$, Sławomir Sta\'nczak$^{1}$}
\IEEEauthorblockA{$^1$\emph{Technische Universität Berlin} and \emph{Fraunhofer Institute for Telecommunications Heinrich-Hertz-Institut}, Berlin, Germany \\
\emph{$^2$KTH Royal Institute of Technology}, Stockholm, Sweden\\
miretti@tu-berlin.de, emilbjo@kth.se, slawomir.stanczak@hhi.fraunhofer.de}
}

\def\baselinestretch{.99}
\setlength{\belowdisplayskip}{1pt}
\setlength{\belowdisplayshortskip}{1pt}
\setlength{\abovedisplayskip}{1pt}
\setlength{\abovedisplayshortskip}{1pt}

\maketitle
\IEEEpubid{\begin{minipage}{\textwidth}\ \\[12pt] \centering
  This work has been submitted to the IEEE for possible publication.\\ Copyright may be transferred without notice, after which this version may no longer be accessible.
\end{minipage}}

\begin{abstract} 
We reconsider the problem of joint power control and beamforming design to maximize the weighted sum rate in large and potentially cell-free massive MIMO networks. In contrast to the available \textit{short-term} methods, where an iterative algorithm is run for every instantaneous channel realization, we derive an iterative algorithm that can be run only sporadically leveraging known channel statistics, with minor performance loss. In addition, our algorithm also applies to the design of non-trivial cooperative beamforming schemes subject to limited sharing of instantaneous channel state information. Furthermore, our algorithm generalizes and outperforms the competing \textit{long-term} methods from the massive MIMO literature, which are restricted to long-term power control only or to long-term joint power control and large-scale fading decoding design.  
\end{abstract}
\begin{IEEEkeywords}
power control, WMMSE, team MMSE, distributed beamforming, XL-MIMO
\end{IEEEkeywords}

\section{Introduction} 
\IEEEpubidadjcol
This study introduces and addresses a novel two-timescale formulation of the classical weighted sum-rate maximization problem in multi-user multiple-input multiple-output (MU-MIMO) wireless systems. By focusing on the uplink of a system with $K$ single-antenna users and a base station equipped with $M$ antennas, and by assuming standard linear receiver processing and single-user coding schemes that treat interference as noise, this classical problem is commonly formulated as
\begin{equation}\label{eq:sumrate}
\begin{aligned}
		\underset{\substack{\vec{p}\in \stdset{R}_{+}^K \\ \vec{v}_1,\ldots,\vec{v}_K \in \stdset{C}^{M}}}{\text{maximize}}
		& \quad\sum_{k\in \set{K}} \omega_k\log(1+\mathsf{SINR}^{\mathsf{inst}}_k(\vec{H},\vec{v}_k,\vec{p})) \\
\text{subject to} & \quad (\forall k\in \set{K})~p_k \leq P,
\end{aligned}
\end{equation}
where $\mathcal{K}=\{1,\ldots,K\}$ is the set of user indexes, $(\omega_1,\ldots,\omega_K)\in \stdset{R}_+^K$ and $P\in \stdset{R}_+$ are problem parameters, and $\mathsf{SINR}^{\mathsf{inst}}_k$ denotes the instanteous SINR of user $k$ for given MU-MIMO channel matrix $\vec{H}\in \stdset{C}^{K\times M}$, power control coefficients $\vec{p}=(p_1,\ldots,p_K)\in \stdset{R}_+^K$, and receive beamformers $\vec{v}_1,\ldots,\vec{v}_K \in \stdset{C}^M$. Although computing an optimal solution to Problem~\eqref{eq:sumrate} is  notoriously challenging, efficient stationary solutions can be obtained using the many iterative algorithms developed in the last two decades \cite{christensen2008wmmse,shi2011wmmse,negro2012sum,shen2018fractional,zhang2023mm}.   These methods apply also to the many variants of \eqref{eq:sumrate} covering, e.g., downlink operations and imperfect channel state information (CSI). 

The main working assumption of this study, which is also ubiquitous in the cellular and cell-free massive MIMO literature \cite{marzetta2016fundamentals,massivemimobook,ngo2017cell,demir2021foundations}, is that the above algorithms for addressing \eqref{eq:sumrate} do not scale well to modern wireless systems, such as those building on the cell-free \cite{ngo2017cell,demir2021foundations} or ultra-large massive MIMO concept \cite{debbah2023xl}. The first challenge is the excessive computational complexity associated with running an iterative algorithm \textit{for every channel realization}. The second challenge is the excessive control signaling overhead for collecting the necessary algorithm input $\vec{H}$ at a central processing unit and for distributing the output $\vec{p}, \vec{v}_1,\ldots,\vec{v}_K$ to the appropriate network entities, \textit{for every channel realization}. Importantly, note that all the above operations must be completed under strict latency constraints, since the optimized powers $\vec{p}$ must be conveyed to the users within the channel coherence time. 
\IEEEpubidadjcol

\subsection{Two-timescale formulation} To mitigate the abovementioned scalability issues of the  approach in \eqref{eq:sumrate}, which we refer to as the \textit{short-term} approach, one possibility is to replace the  instantaneous rate expressions in the objective with  ergodic rate expressions of the type 
\begin{equation}\label{eq:R}
\sum_{k\in \set{K}} \omega_k\E[\log(1+\mathsf{SINR}^{\mathsf{inst}}_k(\rvec{H},\rvec{v}_k,\vec{p})],
\end{equation}
where $\rvec{H}$ is a random matrix, $\vec{p}$ is a deterministic power vector, and the beamformers $\rvec{v}_k$ are \textit{functions} of the available instantaneous CSI (hence, they are random vectors). A key benefit of this two-timescale formulation, which we refer to as the \textit{long-term} approach, is that the powers are optimized sporadically \textit{for many channel realizations}, based on relatively slowly-varying channel statistics. The same benefit applies to the optimization of the beamformers, or, more precisely, of the \textit{functions} mapping instantaneous CSI to beamforming weights. In addition, the long-term approach facilitates the introduction of functional constraints as in \cite{miretti2021team} \cite{miretti2024duality} to enforce distributed beamforming architectures with limited instantaneous CSI sharing, which are particularly relevant for cell-free systems since they allow to distribute the processing load and potentially reduce the fronthaul overhead.

However, the main disadvantage of long-term approaches based on \eqref{eq:R} is that, mostly due to the expectation outside the logarithm, they typically lead to intractable optimization problems. Hence, to derive practical long-term power control algorithms, most studies in the cellular and cell-free massive MIMO literature replace \eqref{eq:R} with the so-called use-and-then-forget (UatF) lower bound on the achievable ergodic rates \cite{marzetta2016fundamentals,massivemimobook,ngo2017cell,demir2021foundations,emil2019lsfd,tran2023maxmin}. For example, by fixing the beamformers, \cite[Sect.~3.4.2]{demir2021foundations} provides an iterative long-term power control algorithm for sum-rate maximization building on the weighted minimum mean-square error (W-MMSE) technique in \cite{shi2011wmmse}. The same technique is applied in \cite{emil2019lsfd} to perform long-term joint power control and large-scale fading decoding (LSFD) design, i.e., the design of a statistical combining stage. 
However, none of the available studies provide long-term joint power control and beamforming design algorithms for sum-rate maximization, i.e., a more general two-timescale variant of \eqref{eq:sumrate}. 

\subsection{Contributions of this study} In this study we close the above gap in the literature by providing a long-term joint power control and beamforming design algorithm for sum-rate maximization. Our algorithm extends the state-of-the-art  in \cite[Algorithm~3.4.2]{demir2021foundations} and \cite{emil2019lsfd} to include beamforming design, i.e., the optimization of the functions mapping instantaneous CSI to beamforming weights. Notably, our algorithm can be applied to both centralized and distributed beamforming architectures. Our derivation builds on the UatF bound and the W-MMSE technique \cite{shi2011wmmse} as in  \cite{demir2021foundations,emil2019lsfd},  and on a recently identified general MMSE-SINR relation for the UatF bound \cite{miretti2021team, miretti2024duality}, reminiscent of the well-known MMSE-SINR relation for \eqref{eq:R} and restricted to centralized beamforming. The distributed beamforming case is covered using the \textit{team MMSE} framework recently developed in \cite{miretti2021team}, where power control was left as an open problem. To illustrate the application of our findings, we provide an updated numerical comparison between short-term and long-term approaches, and between small-cells and cell-free networks.

\textit{Notation:} We denote by $\stdset{R}_+$ the set of nonnegative reals. The $k$th column of the $K$-dimensional identity matrix $\vec{I}_K$ is denoted by $\vec{e}_k$.  The Euclidean norm in $\stdset{C}^K$ is denoted by $\|\cdot\|$. Let $(\Omega,\Sigma,\mathbb{P})$ be a probability space. The expectation and variance operators are denoted, respectively, by $\E[\cdot]$ and $\Var(\cdot)$. We denote by $\set{H}^K$ the set of complex valued random vectors, i.e., $K$-tuples of $\Sigma$-measurable functions $\Omega \to \stdset{C}$ satisfying $(\forall \rvec{x}\in \set{H}^K)$ $\E[\|\rvec{x}\|^2]<\infty$. All equalities involving random quantities should be intended as almost sure (a.s.) equalities. 


\section{System model and problem statement}
We consider the uplink of a large-scale MU-MIMO wireless system composed by $K$ single-antenna users indexed by $\mathcal{K}:=\{1,\ldots,K\}$, and a radio access infrastructure equipped with a total number  of antennas $M$ indexed by $\set{M}:=\{1,\ldots,M\}$. The $M$ infrastructure antennas may be split among many geographically distributed access points equipped with a relatively low number of antennas as in the cell-free massive MIMO concept \cite{ngo2017cell,demir2021foundations}, or colocated as in the ultra-large massive MIMO concept \cite{debbah2023xl}. As customary in the massive MIMO literature, we assume for each time-frequency resource a standard synchronous
narrowband MIMO channel model governed by a stationary ergodic fading process, and simple transmission techniques based on linear receiver processing and on treating interference as noise  \cite{marzetta2016fundamentals,massivemimobook, caire2018ergodic}.

\subsection{Ergodic achievable rates}
For performance evaluation, we consider uplink ergodic achievable rates of the type in \eqref{eq:R}. More specifically, for each user $k \in \set{K}$, we consider the classical expression given by \cite{caire2018ergodic}
\begin{equation}\label{eq:oer}
R_k(\rvec{v}_k,\vec{p}) \eqdef \E\Bigg[\log\Bigg(1+\dfrac{p_k|\rvec{h}_k^\herm\rvec{v}_k|^2}{\underset{j\neq k}{\sum} p_j|\rvec{h}_j^\herm\rvec{v}_k|^2+\|\rvec{v}_k\|_2^2}\Bigg) \Bigg],
\end{equation}
where $\vec{p} \eqdef (p_1,\ldots,p_K) \in \stdset{R}_{+}^K$ collects the transmit powers, $\rvec{h}_k \in\set{H}^M$ is a random vector modeling the (noise normalized, without loss of generality) fading state between user $k$ and all $M$ infrastructure antennas, and $\rvec{v}_k \in \set{H}^M$ models the beamforming vector which is applied by the infrastructure to process the received signals of potentially all $M$ antennas to obtain a soft estimate of the transmit signal of user~$k$. We stress that $\rvec{v}_k$ is generally a function of the instantaneous CSI realizations, and hence it is denoted as a random vector. Since \eqref{eq:oer} is not well-defined for $\rvec{v}_k = \vec{0}$ a.s., we let $(\forall \vec{p} \in \stdset{R}_+^K)$ $R_k(\vec{0},\vec{p})\eqdef 0$. For convenience, we also define the global channel matrix $\rvec{H}\eqdef[\rvec{h}_1,\ldots,\rvec{h}_K]$.

Furthermore, similar to \cite{demir2021foundations,emil2019lsfd}, we consider for optimization purposes the more tractable UatF bound \cite{massivemimobook,caire2018ergodic} $(\forall k \in \set{K})(\forall \vec{p}\in \stdset{R}_+^K)(\forall \rvec{v}_k\in \set{H}^M)$
\begin{equation}\label{eq:uatf}
	R_k^{\mathsf{UatF}}(\rvec{v}_k,\vec{p}) \eqdef \log(1+\mathsf{SINR}_k(\rvec{v}_k,\vec{p})),
\end{equation}
where
\begin{equation*}
	\mathsf{SINR}_k(\rvec{v}_k,\vec{p}) \eqdef \resizebox{0.69\linewidth}{!}{$\dfrac{p_k|\E[\rvec{h}_k^\herm\rvec{v}_k]|^2}{p_k\Var(\rvec{h}_k^\herm\rvec{v}_k)+\underset{j\neq k}{\sum} p_j\E[|\rvec{h}_j^\herm\rvec{v}_k|^2]+\E[\|\rvec{v}_k\|^2]}$}
\end{equation*}
if $\rvec{v}_k \neq \vec{0}$, and $\mathsf{SINR}_k(\vec{0},\vec{p}) \eqdef 0$ otherwise.
\begin{remark}
The UatF bound \eqref{eq:uatf} is a  lower bound on \eqref{eq:oer}, i.e., it satisfies  \cite{caire2018ergodic} $(\forall k \in \set{K})(\forall \vec{p}\in \stdset{R}^K_+)(\forall \rvec{v}_k \in \set{H}^M)$
\begin{equation*}
~R_k^{\mathsf{UatF}}(\rvec{v}_k,\vec{p}) \leq R_k(\rvec{v}_k,\vec{p}).
\end{equation*}
\end{remark}

\subsection{Weighted sum-rate maximization}
Our target is to jointly optimize the uplink transmit powers and the receive beamformers according to a long-term variant of the weighted sum-rate criterion in \eqref{eq:sumrate}. More precisely, under consideration is the following optimization problem:
\begin{equation}\label{eq:wsr}
\begin{aligned}
\underset{\substack{\vec{p}\in \stdset{R}_{+}^K \\ \rvec{v}_1,\ldots,\rvec{v}_K \in \set{H}^{M} }}{\text{maximize}}
& \quad  \sum_{k\in\set{K}}\omega_kR_k^{\mathsf{UatF}}(\rvec{v}_k,\vec{p}) \\
\text{subject to} \; \;  & \quad (\forall k\in \set{K})~p_k \leq P, \\
 & \quad  (\forall k \in \set{K})~\rvec{v}_k \in \set{V}_k,
\end{aligned}
\end{equation}
where $(\omega_1,\ldots,\omega_K)\in \stdset{R}_{+}^K$ is a given vector of positive weights, $R_k^{\mathsf{UatF}}(\rvec{v}_k,\vec{p})$ is given by \eqref{eq:uatf}, and where $\set{V}_k \subseteq \set{H}^M$ denotes a given \textit{information} constraint which can be used to limit the dependency of each entry of $\rvec{v}_k$ to some given (and potentially different) instantaneous CSI. We point out that, if $\set{V}_k = \set{H}^M$, then $\rvec{v}_k$ can be essentially any function of $\rvec{H}$, which is not practical (perfect CSI). Additional details are given in the next section. Note that a similar problem to \eqref{eq:wsr} focusing on the max-min criterion is fully solved in \cite{miretti2022joint,miretti2023fixed}.

\begin{remark}
A well-known challenge in long-term approaches is the need to identify, estimate, and track some form of channel statistics that can be used to produce effective solutions. In contrast, short-term approaches can be often implemented based on minimal statistical information. However, the main working assumption of this study is that short-term approaches are impractical for  large-scale systems, thus making long-term approaches not really an option but rather a necessity.
\end{remark}

\subsection{Information constraints}

We model constraints related to practical centralized and distributed beamforming architectures in a unified way, following the approach introduced in \cite{miretti2021team,miretti2024duality}. We first partition the $M = NL$ infrastructure antennas into $L$ groups of $N$ antennas, and assume that the corresponding entries of $\rvec{v}_k$ are computed by $L$ separate processing units. In the context of cell-free massive MIMO, these processing units can be directly mapped to access points. We further let each user $k\in \set{K}$ be jointly served by a subset $\set{L}_k \subseteq \{1,\ldots,L\}$ of these units. Given an arbitrarily distributed tuple $(\rvec{H},S_1,\ldots,S_L)$, where $S_l$ is the instantaneous CSI at the $l$th processing unit (e.g., noisy measurements of portions of $\rvec{H}$), we then let $(\forall k \in \set{K})$
\begin{equation}\label{eq:distributedCSI}
		\set{V}_k \eqdef \set{V}_{1,k}\times \ldots \times \set{V}_{L,k}, \quad \set{V}_{l,k} \eqdef \begin{cases}
\set{H}_l^N & \text{if } l \in \set{L}_k,\\
\{\vec{0}\} & \text{otherwise},
\end{cases}
\end{equation}
where $\set{H}_l^N\subseteq \set{H}^N$ denotes the set of $N$-tuples of $\Sigma_l$-measurable functions $\Omega \to \stdset{C}^N$ satisfying $(\forall \rvec{x}\in \set{H}_l^N)$ $\E[\|\rvec{x}\|^2]<\infty$, and where $\Sigma_l \subseteq \Sigma$ is the sub-$\sigma$-algebra induced by the CSI $S_l$, which is also called the \emph{information subfield} of the $l$th processing unit. The interested reader is referred to \cite{miretti2024duality} for additional details on the above notions. However, we stress that these notions are by no means required for understanding the key results of this study. The crucial point is that, informally, the constraint set $\set{V}_k$ enforces the subvector of $\rvec{v}_k$ applied by the $l$th processing unit to be a function of $S_l$ only. Importantly, the above constraints cover all common models in the user-centric cell-free massive MIMO literature, including the case of (clustered) centralized beamforming architectures based on global CSI, and the case of (clustered) distributed beamforming architectures based on local CSI \cite{demir2021foundations}.

\section{Approximate solution}
We now derive an iterative algorithm for computing a suboptimal solution to the long-term power control and beamforming design problem \eqref{eq:wsr}. The key technical novelty is the derivation of an equivalent problem formulation which extends the approaches in \cite[Sect.~3.4.2]{demir2021foundations} and \cite{emil2019lsfd} (based on the W-MMSE technique in \cite{shi2011wmmse}) to beamforming design. Our extension is enabled  by a general MMSE-SINR relation recently identified in \cite{miretti2021team,miretti2024duality}, and reported below.
\subsection{SINR maximization via MSE minimization}
Let us consider the MSE between the transmit signal $x_k\sim \CN(0,1)$ of a user~$k\in\set{K}$ and its soft estimate $\hat{x}_k \eqdef \rvec{v}_k^\herm\rvec{y}$ obtained from the received signal $\rvec{y} \eqdef \sum_{k\in \set{K}}\sqrt{p_k}\rvec{h}_kx_k + \rvec{n}$ with noise $\rvec{n}\sim \CN(\vec{0},\vec{I})$, where the noise and all user signals are mutually independent and independent of $(\rvec{H},\rvec{v}_1,\ldots,\rvec{v}_K)$. Specifically, let $(\forall k \in \set{K})(\forall \rvec{v}_k \in \set{V}_k)(\forall \vec{p}\in \stdset{R}_+^K)$ $\mathsf{MSE}_k(\rvec{v}_k,\vec{p})\eqdef$
\begin{equation}\label{eq:MSE}
		\E[|x_k - \hat{x}_k|^2]\\
			=  \E\left[\|\vec{P}^{\frac{1}{2}}\rvec{H}^\herm\rvec{v}_k-\vec{e}_k\|_2^2\right] + \E\left[\|\rvec{v}_k\|^2\right],
\end{equation} 
	where $\vec{P}\eqdef \mathrm{diag}(\vec{p})$, and where the second equality can be verified via simple manipulations. We then have the following.
\begin{proposition}
\label{prop:MMSE-SINR}
For all $k\in\set{K}$ and $\vec{p}\in\stdset{R}_{+}^K$,
\begin{equation*}
1+ \sup_{\rvec{v}_k\in \set{V}_k} \mathsf{SINR}_k(\rvec{v}_k,\vec{p}) = \frac{1}{\inf_{\rvec{v}_k\in \set{V}_k}\mathsf{MSE}_k(\rvec{v}_k,\vec{p})}.
\end{equation*}
Furthermore, $\exists!\rvec{v}_k^\star \in \set{V}_k$ attaining $\inf_{\rvec{v}_k\in \set{V}_k}\mathsf{MSE}_k(\rvec{v}_k,\vec{p})$. Moreover, this $\rvec{v}_k^\star$ also attains $\sup_{\rvec{v}_k\in \set{V}_k} \mathsf{SINR}_k(\rvec{v}_k,\vec{p})$. 
\end{proposition}
\begin{proof} Similar to \cite{miretti2024duality}, and reported in the appendix.
\end{proof}
\begin{remark} Proposition~\ref{prop:MMSE-SINR} should not be confused with the well-known MMSE-SINR relation for instantaneous rate expressions \cite{shi2011wmmse}, such as  the argument of the expectation in \eqref{eq:oer} or in similar expressions covering noisy CSI  \cite[Sect.~4.1]{massivemimobook}, which applies to centralized beamforming architectures only. In contrast, the relation in Proposition~\ref{prop:MMSE-SINR} holds for the UatF bound \eqref{eq:R}, and it applies to arbitrary beamforming architectures, i.e., to arbitrary information constraints $\set{V}_k$.
\end{remark}

\subsection{Equivalent problem formulation}
Equipped with Proposition~\ref{prop:MMSE-SINR}, we now consider the following optimization problem:
\begin{equation}\label{eq:wmmse}
\begin{aligned}
\underset{\substack{\vec{d}\in \stdset{R}_{++}^K,\;\vec{p}\in \stdset{R}_{+}^K\\ \rvec{v}_1,\ldots,\rvec{v}_K \in \set{H}^{M} }}{\text{maximize}}
& \quad  \sum_{k\in\set{K}}\omega_k\left(\log(d_k)-d_k\mathsf{MSE}_k(\rvec{v}_k,\vec{p})\right) \\
\text{subject to} & \quad (\forall k\in \set{K})~p_k \leq P, \\
 & \quad  (\forall k \in \set{K})~\rvec{v}_k \in \set{V}_k,
\end{aligned}
\end{equation}
where $
\mathsf{MSE}_k(\rvec{v}_k,\vec{p})$ is defined in \eqref{eq:MSE}. The above problem is equivalent to Problem~\eqref{eq:wsr}, in the sense specified next.
\begin{proposition}\label{prop:wmmse}
If a solution $(\vec{d}^\star,\vec{p}^\star,\rvec{v}_1^\star,\ldots,\rvec{v}_K^\star)$ to Problem~\eqref{eq:wmmse} exists, then $(\vec{p}^\star,\rvec{v}_1^\star,\ldots,\rvec{v}_K^\star)$ is also a solution to Problem~\eqref{eq:wsr}.
\end{proposition}
\begin{proof}
Assume that $(\vec{d}^\star,\vec{p}^\star,\rvec{v}_1^\star,\ldots,\rvec{v}_K^\star)$ is a solution to Problem~\eqref{eq:wmmse}. First, we notice that $(\forall k \in \set{K})~\mathsf{MSE}_k(\rvec{v}_k^\star,\vec{p}^\star)=\inf_{\rvec{v}_k\in \set{V}_k}\mathsf{MSE}_k(\rvec{v}_k,\vec{p}^\star)$ holds by construction.  Since the objective of~\eqref{eq:wmmse} is concave and continuously differentiable in $\vec{d}$, by the first order optimality condition we obtain
\begin{align*}
(\forall k \in \set{K})~d_k^\star = \frac{1}{\mathsf{MSE}_k(\rvec{v}_k^\star,\vec{p}^\star)} &= \frac{1}{\inf_{\rvec{v}_k\in\mathcal{V}_k}\mathsf{MSE}_k(\rvec{v}_k,\vec{p}^\star)} \\
&= 1 + \sup_{\rvec{v}_k\in \set{V}_k} \mathsf{SINR}_k(\rvec{v}_k,\vec{p}^\star) \\
&= 1+ \mathsf{SINR}_k(\rvec{v}_k^\star,\vec{p}^\star)>0,
\end{align*}
where the last two equalities follow from  Proposition~\ref{prop:MMSE-SINR}.
Hence, the optimum of Problem~\eqref{eq:wmmse} can be rewritten as
\begin{equation*}
\sum_{k\in\set{K}}\omega_k\log(1+ \mathsf{SINR}_k(\rvec{v}_k^\star,\vec{p}^\star))-\sum_{k\in\set{K}}\omega_k.
\end{equation*}
This shows that the optimum of Problem~\eqref{eq:wsr} differs from the optimum of Problem~\eqref{eq:wmmse} only by a constant offset $\sum_{k\in\set{K}}\omega_k$, and that it is attained by  $(\vec{p}^\star, \rvec{v}_1^\star,\ldots,\rvec{v}_K^\star)$. 
\end{proof}

\subsection{Block coordinate ascent algorithm}
Following similar arguments as in \cite{shi2011wmmse,demir2021foundations,emil2019lsfd}, we propose to compute a suboptimal solution to  Problem~\eqref{eq:wsr} using a block coordinate ascent algorithm for the equivalent problem~\eqref{eq:wmmse}. The pseudocode of the proposed algorithm is reported at the top of the next page.

\begin{minipage}{\textwidth}
\begin{algorithmic}[1]
\Require $\vec{p} \in \stdset{R}_{+}^K$
\Repeat
\State $(\forall k \in \set{K})~\rvec{v}_k \gets  \arg\min_{\rvec{v}_k\in\set{V}_k}\mathsf{MSE}_k(\rvec{v}_k,\vec{p})$ 
\State $(\forall k \in \set{K})~d_k \gets \frac{1}{\mathsf{MSE}_k(\rvec{v}_k,\vec{p})}$
\State $(\forall k \in \set{K})~p_k = \min\Big\{ \Big(\frac{\omega_kd_k|\E[\rvec{h}_k^\herm\rvec{v}_k]|}{\sum_{j=1}^K\omega_jd_j\E[|\rvec{h}^\herm_k\rvec{v}_j|^2]}\Big)^2,P\Big\}$
\Until{no significant progress is observed.}
\end{algorithmic}
\vspace{0.1cm}
\end{minipage}
Additional details on the proposed algorithm are given below.
\begin{itemize}[leftmargin=*]
\item The update rule for $(\rvec{v}_1,\ldots,\rvec{v}_K)$ involves solving $K$ 
MSE minimization problems under information constraints \cite{miretti2021team}. For the ideal case $\set{V}_k = \set{H}^M$ (fully centralized processing with perfect CSI), the solution is readily given by the familiar MMSE beamformers $\rvec{v}_k = \left(\rvec{H}\vec{P}\rvec{H}^\herm + \vec{I}\right)^{-1}\rvec{h}_k\sqrt{p_k}$. Similar expressions can also be obtained for many other practical cases, as we will see in the simulation section. We remark once more that $\rvec{v}_k$ is a function of the CSI, and not a deterministic vector. Hence, storing $\rvec{v}_k$ during the algorithm execution means storing its long-term parameters (e.g., the regularization factor in a channel inversion block).
\item The update rule for $\vec{d}$ follows from the first order optimality condition as in the proof of Proposition~\ref{prop:wmmse}.
\item The update rule for $\vec{p}$ is based on the problem
\begin{equation*}
\underset{(\forall k \in \set{K})~0\leq p_k \leq P}{\text{minimize}}\sum_{k\in\set{K}}\omega_kd_k\mathsf{MSE}_k(\rvec{v}_k,\vec{p}).
\end{equation*}
It can be shown that, if the beamformers are obtained from a previous MSE minimization step, then the above problem admits a closed-form solution as in \cite[Eq.~(7.9)]{demir2021foundations}.
\end{itemize}

\noindent Our algorithm enjoys the following monotonicity property:
\begin{proposition}
Let $(\vec{d}^{(i)},\vec{p}^{(i)},\rvec{v}_1^{(i)},\ldots,\rvec{v}_K^{(i)})_{i\in \stdset{N}}$ be the sequence of variables produced by the proposed algorithm. Then, the sequence of objectives $(f^{(i)})_{i\in \stdset{N}}$, where 
\begin{equation*}
(\forall i \in \stdset{N})~f^{(i)}\eqdef \sum_{k\in \set{K}}\omega_k\log(1+\mathsf{SINR}_k(\rvec{v}_k^{(i)},\vec{p}^{(i)})),
\end{equation*}
is monotonically increasing and convergent.
\end{proposition}
\begin{proof}
The proof follows similar arguments as in the proof of Proposition~\ref{prop:wmmse}. It is omitted due to space limitations.
\end{proof}
In fact, the related studies based on \cite{shi2011wmmse} (such as \cite[Sect.~3.4.2]{demir2021foundations} and \cite{emil2019lsfd}) typically advertise convergence not only of the objectives, but also of the variables (to a stationary solution). However, since in this work we are dealing with a functional optimization problem, the notions of convergence and stationarity must be taken with particular care. Hence, we leave additional convergence analysis of our proposed algorithm to an extended version of this study in preparation.  

\begin{figure*}[!ht]
\centering
\subfloat[]{\includegraphics[width=0.33\linewidth]{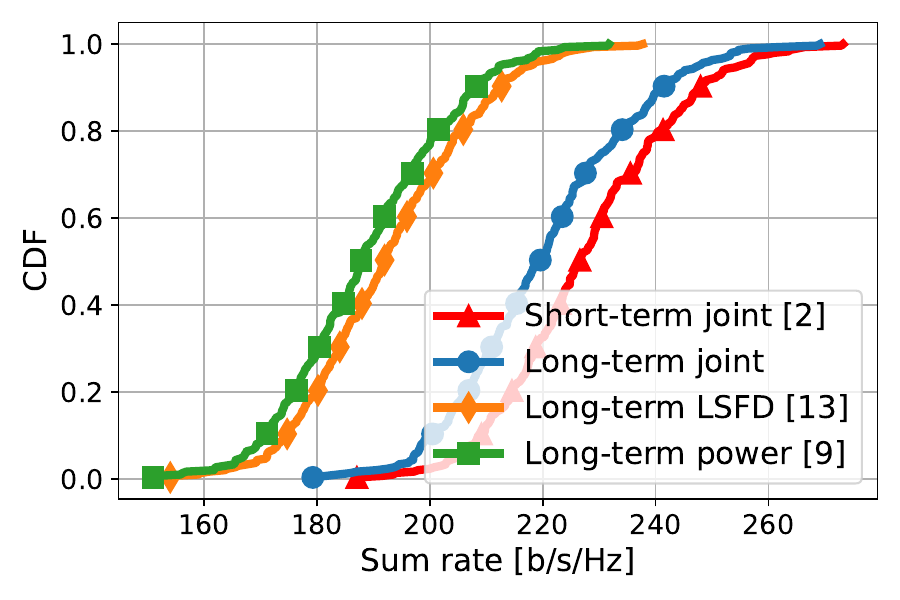}\label{fig:centr}}
\hfil
\subfloat[]{\includegraphics[width=0.33\linewidth]{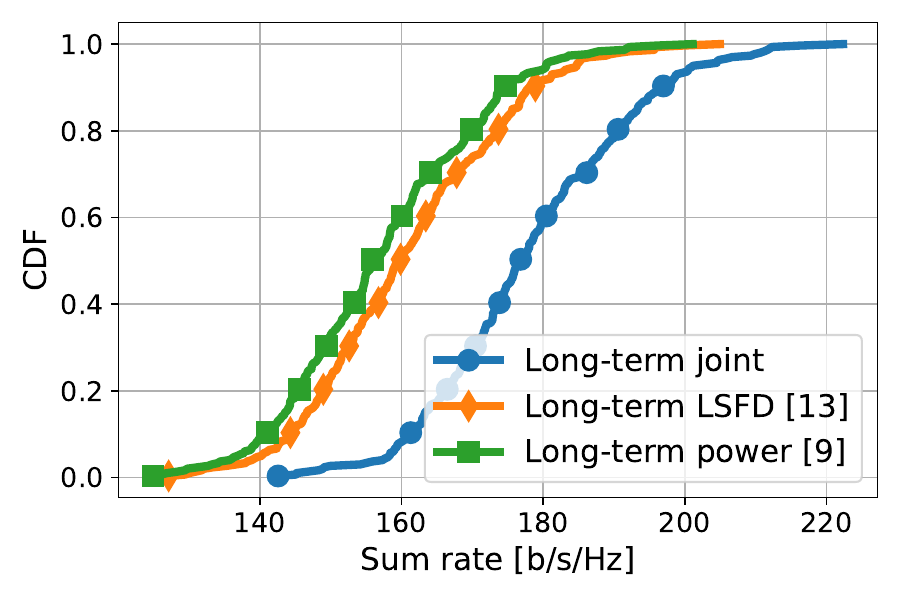}\label{fig:distr}}
\hfil
\subfloat[]{\includegraphics[width=0.33\linewidth]{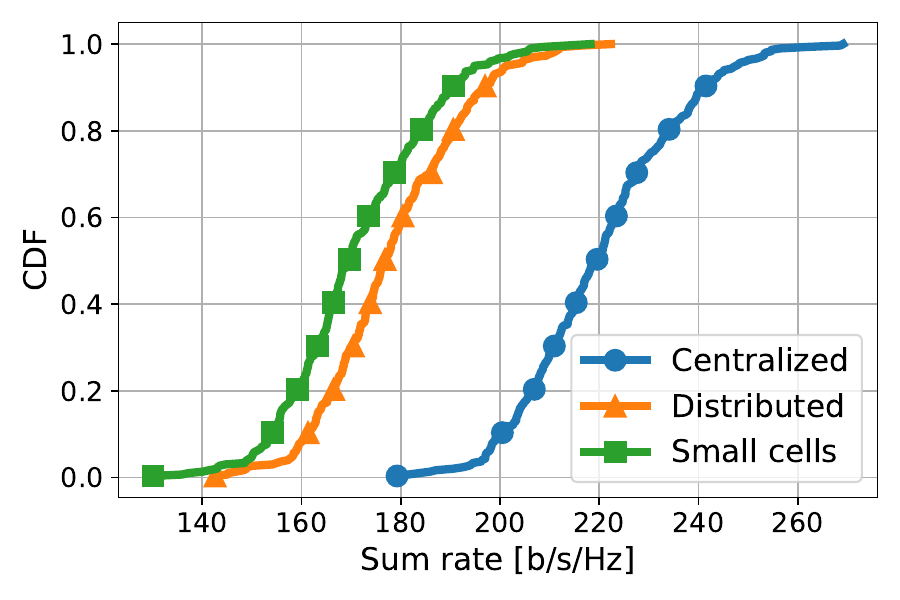}\label{fig:cdf}}
\caption{Comparison of the proposed long-term joint power control and beamforming design method against the competing techniques in \cite{demir2021foundations,emil2019lsfd,shi2011wmmse} for: (a) clustered centralized beamforming (case (i)); and (b) clustered distributed beamforming (case (ii)). Figure (c) compares the three information constraints corresponding to case (i), (ii), and (iii).}
\label{fig:results}
\vspace{-0.5cm}
\end{figure*}
\section{Numerical results and conclusion}
We evaluate numerically the performance of the proposed algorithm over a large MU-MIMO system composed by $L=16$ regularly spaced access points with $N=4$ antennas each, serving in the same time-frequency resource $K=64$ uniformly distributed single-antenna users within a squared service area of size $500\times 500~\text{m}^2$. Note that, in contrast to the standard (cell-free) massive MIMO regime with $M= NL \gg K$, here we consider a critically loaded system with $M \approx K$. We use the same simulation environment as in \cite{miretti2023fixed,miretti2024duality}, which in turn uses the same channel model as in \cite{demir2021foundations} for a system with 2 GHz carrier frequency and 20 MHz bandwidth. 

We focus on  Problem~\eqref{eq:wsr} with unitary weights and a per-user power budget $P=20$~dBm. As information constraints, we consider three common cases as in \cite{ngo2017cell,demir2021foundations,miretti2024duality,miretti2023fixed}: (i)~centralized cell-free network with user-centric clustering; (ii)~distributed cell-free network with user-centric clustering; (iii)~small-cells network. More specifically, we let each user be jointly served by its $Q$ strongest access points, where $Q=4$ for case (i) and (ii), and $Q=1$ for case (iii). We further consider a simple  channel acquisition model that reflects the canonical cell-free massive MIMO implementation with pilot-based local channel estimation, and possible CSI sharing through the fronthaul. In all cases, we neglect for simplicity local estimation errors, and we assume a local CSI model where each access point acquires only the instantaneous channels between itself and its associated users. In case (i), the beamformers are computed based on full CSI sharing within the corresponding user-centric cluster of access points. In case (ii) and (iii), the beamformers are computed locally by the access points based on local CSI only. The corresponding optimal beamformers, i.e., attaining $\inf_{\rvec{v}_k\in\set{V}_k}\mathsf{MSE}_k(\rvec{v}_k,\vec{p})$, are given in closed-form by: (i) the centralized MMSE solution \cite{demir2021foundations} \cite[Proposition~9]{miretti2024duality}; (ii) the local \textit{team} MMSE solution \cite{miretti2021team} \cite[Proposition~11]{miretti2024duality}; and (iii) the local MMSE solution \cite{demir2021foundations}. Due to space limitations, we refer to \cite{miretti2024duality,miretti2023fixed} for additional technical details on the chosen $\set{V}_k$ and corresponding solutions.

In Fig.~\ref{fig:centr} and Fig.~\ref{fig:distr}, we plot the empirical cumulative density function (CDF) of the ergodic sum-rate in \eqref{eq:oer} for $300$ independent user drops, achieved by the proposed algorithm and by: long-term power control \cite{demir2021foundations}; long-term joint power control and LSFD design \cite{emil2019lsfd}; short-term joint power control and beamforming design \cite{shi2011wmmse} (not applicable to case (ii)). For the methods in \cite{demir2021foundations} and \cite{emil2019lsfd}, we fix the beamformers to the aforementioned optimal beamformers assuming full power $\vec{p}=\vec{1}P$, as in \cite{demir2021foundations}. Our main conclusions are: 
\begin{itemize}[leftmargin=*]
\item (Fig.~\ref{fig:centr}) As expected (and in contrast to the results in \cite{miretti2023fixed} for the max-min criterion), the short-term approach \cite{shi2011wmmse} performs best. However, the gains may not justify the higher complexity and signaling overhead, which, in fact, pose significant scalability issues as discussed in the introduction;
\item (Fig.~\ref{fig:centr} and Fig.~\ref{fig:distr}) The proposed long-term method significantly outperforms the competing long-term methods \cite{demir2021foundations,emil2019lsfd}. The reason is that $\vec{p}=\vec{1}P$ is highly suboptimal in a critically loaded system, since maximizing the sum-rate implies that only a limited number of strong users are jointly served in the same time-frequency resource. 
\end{itemize}
Finally, Fig.~\ref{fig:cdf} reports the performance of the proposed method for all the considered information constraints. As expected, centralized cell-free networks may provide much larger sum-rates. Perhaps more surprisingly, and differently than for the significant max-min fairness gains reported in \cite{miretti2023fixed}, distributed cell-free architectures seem to provide marginal sum-rate gains over small cells. We remark that we observed this behavior also under different simulation parameters. 

An informal reason it that, compared to small cells, distributed cell-free networks mostly provide a more uniform SNR, but not much better interference management capabilities. However, uniform SNR is mostly beneficial for cell-edge users, which are usually not scheduled for service when maximizing the sum-rate. Nevertheless, we point out that additional investigation is needed to corroborate the above conclusion, since the proposed algorithm is suboptimal, and since there may be other scenarios giving different results.
	
\appendix[Proof of Proposition~\ref{prop:MMSE-SINR}]
\label{app:MMSE-SINR}
We first observe that $\inf_{\rvec{v}_k\in\set{V}_k} \mathsf{MSE}_k(\rvec{v}_k,\vec{p})$  
\begin{align*}
			&\overset{(a)}{=} \inf_{\rvec{v}_k\in\set{V}_k\backslash \{  \vec{0} \}} \inf_{\beta \in \stdset{C}} \mathsf{MSE}_k(\beta\rvec{v}_k,\vec{p}) \\
			&\overset{(b)}{=} \inf_{\rvec{v}_k\in\set{V}_k\backslash \{  \vec{0} \}} 1-\dfrac{p_k|\E[\rvec{h}_k^\herm\rvec{v}_k]|^2}{\sum_{j\in\set{K}}p_j\E[|\rvec{h}_j^\herm\rvec{v}_k|^2]+\E\left[\|\rvec{v}_k\|^2\right]} \\
			&= \dfrac{1}{1+\sup_{\rvec{v}_k\in\set{V}_k\backslash \{  \vec{0} \}}\mathsf{SINR}_k(\rvec{v}_k,\vec{p})} \\
&\overset{(c)}{=} \dfrac{1}{1+\max\bigg\{\underset{\rvec{v}_k\in \{\vec{0}\}}{\sup}\mathsf{SINR}_k(\rvec{v}_k,\vec{p}),\underset{\rvec{v}_k\in \set{V}_k\backslash\{\vec{0}\}}{\sup}\mathsf{SINR}_k(\rvec{v}_k,\vec{p})\bigg\}} \\
&\overset{(d)}{=} \dfrac{1}{1+\sup_{\rvec{v}_k\in\set{V}_k}\mathsf{SINR}_k(\rvec{v}_k,\vec{p})} 
\end{align*}
where $(a)$ follows from $(\forall \beta\in \stdset{C})(\forall \rvec{v}_k \in \mathcal{V}_k)$ $\beta \rvec{v}_k \in\mathcal{V}_k$, $(b)$ follows from standard minimization of scalar quadratic forms, $(c)$ follows from the nonnegativity of the SINR, and $(d)$ follows from an elementary identity on the supremum of the union of bounded sets. Existence and uniqueness of $\rvec{v}_k^\star \in \set{V}_k$ attaining $\inf_{\rvec{v}_k\in \set{V}_k}\mathsf{MSE}_k(\rvec{v}_k,\vec{p})$ follows from the Hilbert projection theorem  as in \cite[Lemma~4]{miretti2021team}. It remains to verify that this $\rvec{v}_k^\star$ also attains $\sup_{\rvec{v}_k\in \set{V}_k} \mathsf{SINR}_k(\rvec{v}_k,\vec{p})$. If $\rvec{v}_k^\star \neq \vec{0}$ a.s., this can be verified by following similar steps as in the above chain of equalities: $
\mathsf{MSE}_k(\rvec{v}_k^\star,\vec{p}) =  \inf_{\beta \in \stdset{C}} \mathsf{MSE}_k(\beta\rvec{v}_k^\star,\vec{p}) = (1+\mathsf{SINR}_k(\rvec{v}_k^\star,\vec{p}))^{-1}$. Finally, the particular case $\rvec{v}_k^\star = \vec{0}$ a.s. can be verified by direct inspection.  

\bibliographystyle{IEEEbib}
\bibliography{IEEEabrv,refs}

\end{document}